\documentclass[a4paper,11pt]{article}

\usepackage{graphicx}

\usepackage{soul}

\usepackage{fullpage}
\usepackage{libertine}
\usepackage{color}

\usepackage[ruled,vlined]{algorithm2e}
\SetArgSty{textrm}

\usepackage{amsmath,amsfonts,amssymb,amsthm}

\usepackage[breaklinks=true]{hyperref}
\usepackage[svgnames]{xcolor}
\usepackage[capitalise,nameinlink]{cleveref}
\hypersetup{colorlinks={true},linkcolor={DarkBlue},citecolor=[named]{DarkGreen}}

\usepackage{natbib}

\usepackage{subcaption}

\usepackage{tikz}  
\usetikzlibrary{arrows}
\usetikzlibrary{patterns,snakes}
\usetikzlibrary{decorations.shapes}
\tikzstyle{overbrace text style}=[font=\tiny, above, pos=.5, yshift=5pt]
\tikzstyle{overbrace style}=[decorate,decoration={brace,raise=5pt,amplitude=3pt}]
\usetikzlibrary{shapes.geometric}

\newtheorem{theorem}{Theorem}[section]
\newtheorem{corollary}[theorem]{Corollary}

\theoremstyle{definition}
\newtheorem*{comment*}{Comment}

\newcommand{\SW}{\text{SW}}

\newcommand{\bw}{\mathbf{w}}
\newcommand{\bo}{\mathbf{o}}

\newcommand{\bx}{\mathbf{x}}

\newcommand{\bp}{\mathbf{p}}
\newcommand{\balpha}{\boldsymbol{\alpha}}
\newcommand{\opt}{\text{OPT}}
\newcommand{\mech}{\text{MECH}}

\title{\bf Truthful Facility Location with Candidate Locations and Limited Resources}

\usepackage{authblk}
\author{Panagiotis Kanellopoulos and Alexandros A. Voudouris}
\date{School of Computer Science and Electronic Engineering \\ University of Essex, UK}

\begin{document}

\allowdisplaybreaks

\maketitle

\begin{abstract}
We study a truthful facility location problem where one out of $k\geq2$ available facilities must be built at a location chosen from a set of candidate ones in the interval $[0,1]$. This decision aims to accommodate a set of agents with private positions in $[0,1]$ and approval preferences over the facilities; the agents act strategically and may misreport their private information to maximize their utility, which depends on the chosen facility and their distance from it. We focus on strategyproof mechanisms that incentivize the agents to act truthfully and bound the best possible approximation of the optimal social welfare (the total utility of the agents) they can achieve. We first show that deterministic mechanisms have unbounded approximation ratio, and then present a randomized mechanism with approximation ratio $k$, which is tight even when agents may only misreport their positions. For the restricted setting where agents may only misreport their approval preferences, we design a deterministic mechanism with approximation ratio of roughly $2.325$, and establish lower bounds of $3/2$ and $6/5$ for deterministic and randomized mechanisms, respectively. 
\end{abstract}

\section{Introduction} \label{sec:intro}
Facility location problems, where the objective is to decide where in a metric space to build a number of facilities to accommodate the needs of a set of individuals, form a paradigmatic class of challenges that has attracted interest across multiple fields, including Theoretical Computer Science, Artificial Intelligence, Multi-Agent Systems, and Operations Research. In many real-world settings, the input to such problems is provided by self-interested agents who have private information (e.g., their locations) and may misreport it if doing so leads to outcomes that benefit them. This has led to significant interest in the design of \emph{strategyproof} facility location mechanisms that incentivize agents to report their actual private information. The framework of \emph{approximate mechanism design without money} \citep{procaccia09approximate} has played a key role in this area, giving rise to a rich collection of models that capture practical applications. 

In the classic setting studied by \citet{procaccia09approximate}, the agents are positioned on the real line and a set of homogeneous  facilities (that offer the same type of service) can all be built without restrictions anywhere on the line. In many practical settings, however, there may be different types of facilities which cannot all be built due to budgetary and planning constraints. Consider, for instance, that a city has sufficient funds to build only one hospital and must decide an appropriate specialization (e.g., pediatrics, trauma, geriatrics) at a location among a set of prespecified ones. City residents report their home addresses and preferred hospital type, but may strategically misreport this private information to increase the likelihood of an approved specialization being built at a nearby location. The central question then becomes: how can the city make an (approximately) optimal decision while ensuring truthfulness?

Motivated by applications like the hospital one described above, \citet{deligkas2023limited} introduced a model with strategic agents that are positioned in the interval $[0,1]$ and have {\em approval} (also known as optional) {\em preferences} over a set of \emph{heterogeneous} facilities that offer different types of services. Given information reported by the agents about their positions and preferences, one of the available facilities is chosen to be built at some location in the interval $[0,1]$, and the agents derive a utility that is a function of the chosen facility as well as its location. In this work, we consider an extension of this model in which the chosen facility cannot be built anywhere in $[0,1]$ but only at a location chosen from a predetermined set of \emph{candidate} ones; such candidate locations naturally capture planning (or zoning) constraints that limit the set of available locations. This model elegantly combines elements of traditional facility location, as agents care about proximity to the facility, with voting theory, as agents care about which facility is chosen. 

\subsection{Our Contribution}
To be more specific, in our model, there are $k\geq 2$ facilities, and a set of $n$ agents with private positions in the interval $[0,1]$ and approval preferences ($0$ or $1$) for the facilities. Given input from the agents about their private information, the goal is to compute a solution consisting of one of the facilities and a location chosen from a set of candidate ones where the facility will be built. The agents act strategically aiming to maximize their utility, which is either $0$ in case they do not approve the chosen facility, or depends on their distance from the facility location if they do approve it. Our goal is to design mechanisms that incentivize the agents to truthfully report their private information, and approximate the optimal social welfare (the total utility of all agents). 

We start with the general setting in Section~\ref{sec:general}, where the agents may misreport their positions and their preferences. Using a simple construction, we show that no deterministic strategyproof mechanism can achieve a bounded approximation ratio. This is in sharp contrast to the approximation ratio of $2$ that can be achieved by deterministic mechanisms in the continuous model of \citet{deligkas2023limited}. To overcome this impossibility, we then turn to randomization, and show that it is possible to achieve a tight bound of $k$ via a mechanism which defines different probability distributions depending on which subinterval of $[0,1]$ contains the candidate locations. All these results for the general setting are also true even when the preferences of the agents are assumed to be known and thus the agents can only misreport their positions. 

In Section~\ref{sec:positions}, we consider a restricted setting, where the positions of the agents are known and thus the agents may only misreport their preferences. We show that a small constant approximation ratio of roughly $2.325$ can be achieved via a deterministic strategyproof mechanism. As in the case of randomization in the general setting, this mechanism also distinguishes between different cases depending on how the candidate locations are distributed in the interval $[0,1]$, and makes decisions taking advantage of the fact that the positions of the agents cannot be manipulated; in particular, once we fix a candidate location, we can then simply choose the optimal facility for this location without providing incentives to the agents to misreport.

\subsection{Related Work}
The framework of approximate mechanism design without money, introduced by \citet{procaccia09approximate}, initiated a broad line of research on truthful facility location, expanding on classical characterizations such as that of \citet{Moulin1980}; see \citet{fl-survey} for a comprehensive overview of work in this area. We focus here on works that are most relevant to us on the following axes: heterogeneous facilities, approval preferences, and location constraints. 

As already mentioned above, the work closest to ours is that of \citet{deligkas2023limited} whose model allowed the chosen facility to be built at any point in the interval $[0,1]$. \citeauthor{deligkas2023limited} showed bounds mainly for $k=2$ and three settings, depending on whether the agents can misreport their positions, their approval preferences, or both. In all cases, their results are small constant bounds on the approximation ratio of deterministic and randomized mechanisms. In follow-up work, \citet{FangL24} modified the model of \citet{deligkas2023limited} by assigning to each agent an approval weight of $1$ which can be distributed to the facilities, leading to fractional preferences. They also study deterministic and randomized mechanisms, and present small constant bounds for the general, the known-positions and the known-preferences settings. 

In a similar spirit, \citet{Elkind2022approval} studied a multiwinner facility location problem on the line, where $k$ out of $m$ facilities are to be built. As in our case, their model assumes the existence of candidate locations and approval preferences. However, in contrast to us, the preferences of the agents depend on the distance from facility locations. In particular, each agent is associated with a radius per facility, and approves a facility only if it is chosen and built within the corresponding radius. Among other results, they present fairness-related axioms that are motivated by multiwinner voting; similar connections of facility location and voting have also been highlighted by \citet{feldman2016voting} and in various works within the distortion literature~\citep{distortion-survey}.

Our work is also related to a long list of papers which study models with multiple facilities and sufficient funds to build all of them. Similarly to us, most of these papers consider agents that have approval preferences over the facilities\footnote{Other types of preferences have also been studied, such as hybrid~\citep{feigenbaum2015hybrid,anastasiadis2018heterogeneous} and fractional~\citep{fong2018fractional}.} and are positioned on a line (either that of real numbers or a graph). Some of them allow the facilities to be built at any point of the underlying line~\citep{chen2020max,li2020constant,Sha_Bao_Chan_Chau_Fong_Li_2025,serafino2016, kanellopoulos2023discrete}, while others, as in our case, focus on candidate locations~\citep{Tang2020candidate,Zhao2023constrained,lotfi2023max,ZhaoLNF24,kanellopoulos2024}. Models involving minimum distance between the facilities or dynamic locations, instead of explicit candidate locations, have also been proposed~\citep{Xu2021minimum,Duan2021minimum,DLV24}.

\section{The Model} \label{sec:prelim}
There is a set $N$ of $n \geq 2$ {\em agents}, and a set of $k \geq 2$ {\em facilities} called $F_1, \ldots, F_k$. Every agent $i \in N$ has a {\em position} $x_i \in [0,1]$ and an {\em approval preference} $\alpha_{i,j} \in \{0,1\}$ for each facility $F_j$ such that $\alpha_{i,j} = 1$ indicates that $i$ approves $F_j$ and $\alpha_{i,j}=0$ indicates that $i$ does not approve $F_j$. For any $j \in [k]$, let $N_j$ be the set of agents that approve facility $F_j$, and $n_j = |N_j|$; note that agents might belong to multiple such sets in case they approve multiple facilities. 

The objective is to choose one of the facilities and build it at a location chosen from a set $C$ of candidate locations in $[0,1]$. In particular, a feasible {\em solution} is a pair $(j,x)$ with $j \in [k]$ and $x \in C$. 
A randomized solution $\bp$ is a probability distribution that assigns a probability $p_{j,x}$ to any solution $(j,x)$; note that any deterministic solution $(j,x)$ is a randomized solution such that $p_{j,x} = 1$. 
Given $\bp$, each agent $i$ derives an (expected) {\em utility} equal to 
\begin{align*}
    u_i(\bp) = \sum_{j\in [k]} \sum_{x \in C} p_{j,x} \cdot \alpha_{i,j} \cdot \bigg( 1 - d(i,x) \bigg),
\end{align*}
where $d(i,x) = |x_i-x|$ is the distance between the position $x_i$ of agent $i$ and $x$. The (expected) {\em social welfare} of a randomized solution $\bp$ is the total utility of the agents for it: 
\begin{align*}
    \SW(\bp) = \sum_i u_i(\bp).
\end{align*}

A {\em mechanism} $M$ takes as input an instance $I = (\bx,\balpha,C)$ consisting of the reported positions $\bx$ and preferences $\balpha$ of the agents, as well as the candidate locations $C$, and outputs a probability distribution $\bp_M(I)$ over feasible solutions. Our goal is to design mechanisms that are strategyproof and achieve a good approximation of the maximum possible social welfare. A mechanism is said to be {\em strategyproof} (in expectation) if no agent can misreport its private information and lead to an outcome that increases its expected utility. 
In particular, a mechanism is strategyproof if, for any two instances $I = (\bx,\balpha,C)$ and $I' = ((x_i',\bx_{-i}),(\alpha_i',\balpha_{-i}),C)$ which differ only on the private information (position or preferences) reported by a single agent $i$, $u_i(\bp_M(I)) \geq u_i(\bp_M(I'))$.

The {\em approximation ratio} of a mechanism $M$ is defined as the worst-case (over all possible instances) ratio between the maximum possible social welfare achieved over any feasible solution and the social welfare achieved by the solution computed by the mechanism, that is, 
\begin{align*}
    \sup_{I} \frac{\max_{(j,x)} \SW(j,x)}{\SW(\bp_M(I))}.
\end{align*}
For simplicity, when an instance is clear from context, we will write $\opt$ to refer to the optimal social welfare and $\mech$ to refer to the social welfare of the mechanism. Then, the approximation ratio is simply the ratio $\opt/\mech$.

\section{General Setting} \label{sec:general}
In this section we consider the general setting where the agents can in principle misreport about their positions and preferences. We start by showing a strong impossibility result for deterministic strategyproof mechanisms, which holds even when the preferences of the agents are assumed to be known.  

\begin{theorem}\label{thm:general:lower:deterministic}
For any $k \geq 2$, the approximation ratio of any deterministic strategyproof mechanism is unbounded, even when the preferences of the agents are known.
\end{theorem}

\begin{proof}
We present the proof for $k=2$; extending it to the case $k\geq 2$ is straightforward. 
Consider an instance $I$ with two agents $i$ and $j$ such that both are positioned at an infinitesimal $\varepsilon > 0$, $i$ approves only $F_1$, and $j$ approves only $F_2$. There is a single candidate location at $1$. A deterministic mechanism will choose one of the two facilities and place it at the only possible location. Without loss of generality, suppose that $F_1$ is placed at $1$. Consequently, agent $j$, who approves $F_2$, derives a utility of $0$. 

Now consider another instance $J$ which is the same as $I$ with the only difference that agent $j$ has been moved to $1$. If the mechanism chooses to place $F_2$ at $1$ in $J$, then when the true position of agent $j$ is $\varepsilon$ as in $I$, $j$ would prefer to misreport its position to be $1$ as in $J$ so that $F_2$ is chosen, and $j$ increases its utility from $0$ to $\varepsilon$. Hence, for the mechanism to be strategyproof, it has to be the case that $F_1$ is chosen in $J$ as well. But then, since the social welfare of choosing $F_1$ is $\varepsilon$ and the social welfare of choosing $F_2$ is $1$, the approximation ratio is unbounded. 
\end{proof}

Exploiting randomization, we can achieve a bounded, tight approximation ratio of $k$, for any $k \geq 2$. We first present the lower bound, which holds for any randomized mechanism and even in the case of known preferences. 

\begin{theorem}\label{thm:general:lower:rand}
For any $k \geq 2$, the approximation ratio of any randomized strategyproof mechanism is at least $k$, even when the preferences of the agents are known. 
\end{theorem}

\begin{proof}
Consider an instance $I$ with $k$ agents positioned at an infinitesimal $\varepsilon > 0$ and a single candidate location at $1$. Each agent approves exactly one facility so that each facility is approved by an agent. Consider any randomized strategyproof mechanism. By definition, there exists a facility $F_i$ that the mechanism chooses to place with probability $p_i \leq 1/k$. 

Consider now another instance $J$ which is the same as $I$ with the only difference that the agent $i$ who approves $F_i$ has been moved to $1$. Since the mechanism strategyproof, it cannot choose $F_i$ with higher probability than in $I$; otherwise, when $I$ is the true instance, agent $i$ would have an incentive to misreport its position as $1$, leading to instance $J$, and increasing the probability to obtain positive utility. Therefore, the utility of agent $i$ is at most $1/k$ in $J$, and thus $\mech \leq 1/k + (k-1)\varepsilon/k$. In contrast, the optimal solution is to deterministically choose facility $F_i$ for a social welfare of $\opt = 1$, implying the lower bound of $k$ as $\varepsilon$ tends to zero. 
\end{proof}

We now focus on showing a matching upper bound for any $k \geq 2$. Our mechanism distinguishes between differences cases depending on where the leftmost and rightmost candidate locations are in the interval $[0,1]$. Without loss of generality, let $F_1$ be the facility that is approved by most agents (i.e., $n_1 = \max_{j \in [k]} n_j$), $L$ the leftmost candidate location, and $R$ the rightmost candidate location. We consider the following exhaustive cases:
\begin{itemize}
    \item (Case 1) There is a single candidate location: We choose each facility equiprobably. This clearly leads to a $k$-approximation since the optimal solution is chosen with probability $1/k$. 
    
    \item (Case 2) There is a candidate location in the interval $[1/k,(k-1)/k]$: We deterministically place $F_1$ at this location.

    \item (Case 3) $L < 1/k < (k-1)/k < R$: We place $F_1$ either at $L$ or at $R$ with appropriately defined probabilities.

    \item (Case 4) $(k-1)/k < L < R$: We place each facility $F_j$, $j\in[k]$ at $L$ with appropriately defined probabilities.

    \item (Case 5) $L < R < 1/k$: We place each facility $F_j$, $j\in[k]$ at $R$ with appropriately defined probabilities.
\end{itemize}
See Mechanism~\ref{mech:general:upper:k} for the formal definition of the mechanism and the corresponding probability distribution for each case. Note that (Case 4) and (Case 5) are symmetric. 

\newcommand\mycommfont[1]{\normalfont\textcolor{blue}{#1}}
\SetCommentSty{mycommfont}
\begin{algorithm}[h]
\SetNoFillComment
\caption{}
\label{mech:general:upper:k}
{\bf Input:} $k \geq 2$, candidate locations, reported agent positions and preferences\;
{\bf Output:} Probability distribution $\mathbf{p}$ over feasible solutions\;
Rename facilities such that $n_1 = \max_{j\in[k]} n_j$\;
$L \gets$ leftmost candidate location\;
$R \gets$ rightmost candidate location\;
\tcp*[h]{(Case 1)} \\
\uIf{$L=R=X$}
{
   \For{$j \in \{1, \ldots, k\}$}
    {
        $p_{j,L} \gets 1/k$\;
    } 
}
\tcp*[h]{(Case 2)} \\
\uElseIf{$\exists$ candidate location $X \in [1/k, (k-1)/k]$}
{
    $p_{1,X} \gets 1$\;
}
\tcp*[h]{(Case 3)} \\
\uElseIf{$L < 1/k < (k-1)/k < R$}{
   $p_{1,L} \gets \frac{1-k+kR}{k(R-L)} $\;
   $p_{1,R} \gets 1-p_{1,L}$\; 
}
\tcp*[h]{(Case 4)} \\
\uElseIf{$(k-1)/k < L < R$}{
    $p_{1,L} \gets \frac{n_1-\frac{k}{k-1}(1-L)(n-n_1)}{kLn_1-\frac{k}{k-1}(1-L)(n-n_1)}$\;
    \For{$j \in \{2, \ldots, k\}$}
    {
        $p_{j,L} \gets \frac{1-p_{1,L}}{k-1}$\;
    }
}
\tcp*[h]{(Case 5)} \\
\uElseIf{$L < R < 1/k$}{
    $p_{1,R} \gets  \frac{n_1-\frac{k}{k-1}R(n-n_1)}{k(1-R)n_1-\frac{k}{k-1}R(n-n_1)}$\;
    \For{$j \in \{2, \ldots, k\}$}
    {
        $p_{j,R} \gets \frac{1-p_{1,R}}{k-1}$\;
    }
}   
\Return $\mathbf{p}$\;
\end{algorithm}

\begin{theorem} \label{thm:general:randomized:upper:k:sp}
Mechanism~\ref{mech:general:upper:k} is strategyproof.
\end{theorem}

\begin{proof}
Observe that all the cases considered by the mechanism depend only on known information, in particular, the parameters $k$, $L$ and $R$. Hence, any possible misreport by the agents about their positions or preferences cannot affect the case under consideration. The probability distribution defined in each case does not depend on the positions of the agents in any case, and thus the agents have no incentive to misreport their positions overall. In (Case 1), (Case 2) and (Case 3), the probability distribution is a function of known information, and thus the agents clearly have no incentive to misreport their preferences. In (Case 4) and (Case 5), the location is fixed and the probability distribution is a function that is increasing in the number $n_1$ of agents that approve $F_1$ (which is the most approved facility). Let $X$ be the facility location ($L$ for (Case 4) and $R$ for (Case 5)) and consider an arbitrary agent $i$.
\begin{itemize}
\item
$i$ does not truly approve $F_1$. $i$ can only affect the outcome of the mechanism by misreporting that it approves $F_1$, which can increase the probability of $F_1$ and decrease the probability of placing facilities that $i$ truly approves, leading to a possible decrease in expected utility. 

\item 
$i$ truly approves $F_1$. By definition, the expected utility of $i$ is 
\begin{align*}
    u_i(\bp) 
    &= \sum_{j \in [k]} p_{j,X} \cdot \alpha_{i,j} \cdot \bigg( 1-d(i,X) \bigg) \\
    &= \bigg( 1-d(i,X) \bigg) \bigg( 1 - \frac{1}{k-1} \sum_{j \neq 1}\alpha_{i,j}\bigg) \cdot p_{1,X} + \frac{1-d(i,X)}{k-1} \sum_{j \neq 1}\alpha_{i,j}.
\end{align*}
Since $\sum_{j \neq 1}\alpha_{i,j} \leq k-1$, $u_i(\bp)$ is an increasing function in $p_{1,X}$. Hence, if $i$ misreports that it does not approve $F_1$, $p_{1,X}$ may decrease, which can lead to a decrease in expected utility.  
\end{itemize}
Consequently, no agent has any incentive to misreport their preferences in (Case 4) and (Case 5) as well, and thus the mechanism is strategyproof overall. 
\end{proof}

We next focus on bounding the approximation ratio.

\begin{theorem} \label{thm:general:randomized:upper:k:ratio}
The approximation ratio of Mechanism~\ref{mech:general:upper:k} is at most $k$, for any $k \geq 2$.
\end{theorem}

\begin{proof}
Since (Case 1) directly leads to a $k$-approximation (the optimal solution is chosen with probability $1/k$), and (Case 5) is symmetric to (Case 4), it suffices to argue about the following three cases.
\begin{itemize}
    \item (Case 2) There is a candidate location $X\in [1/k, (k-1)/k]$;
    \item (Case 3) $L < \frac{1}{k} < \frac{k-1}{k} < R$;
    \item (Case 4) $\frac{k-1}{k} < L < R \leq 1$.
\end{itemize}
Before considering each case, recall that $n_1 = \max_j n_j$. 

\medskip
\noindent
{\bf (Case 2): There is a candidate location $X \in [\frac{1}{k}, \frac{k-1}{k}]$.} 
In this case, the mechanism places deterministically $F_1$ at $X$. 
Consider an arbitrary agent $i \in N_1$.
\begin{itemize}
    \item If $x_i \leq 1/k$, then $d(i,X) \leq \frac{k-1}{k}$.

    \item If $x_i \geq (k-1)/k$, then $d(i,X) \leq 1-\frac{1}{k} = \frac{k-1}{k}$.

    \item If $1/k < x_i < (k-1)/k$, then $d(i,X) \leq \frac{k-1}{k} - \frac{1}{k} = \frac{k-2}{k} \leq \frac{k-1}{k}$.
\end{itemize}
Hence, in any case, $u_i(\bp) = 1 - d(i,X) \geq 1/k$, which implies that $\mech \geq n_1/k$. Since $\opt \leq \max_j n_j =n_1$, the approximation ratio is at most $k$.

\medskip
\noindent
{\bf (Case 3): $L < \frac{1}{k} < \frac{k-1}{k} < R$.} 
We will show that the expected utility of any agent $i \in N_1$ is always at least $1/k$ by switching between cases depending on the positions of such agents relative to $L$ and $R$. Recall that the mechanism places $F_1$ at $L$ with probability $p_{1,L} = \frac{1-k+kR}{k(R-L)}$ and at $R$ with the remaining probability $1-p_{1,L}$. 
\begin{itemize}
\item If $x_i \leq L$, since $d(i,L) \leq L$ and $d(i,R) \leq R$, using the definition of $p_{1,L}$, we have
\begin{align*}
    u_i(\bp)
    &= p_{1,L} \cdot \big(1-d(i,L)\big) + \big(1-p_{1,L}\big)\cdot \big(1-d(i,R)\big) \\
    &\geq p_{1,L} \cdot  (1 - L) + \big(1-p_{1,L}\big)\cdot (1 - R) \\
    &= \frac{(1-k+kR)(1-L) + (k-1-kL)(1-R) }{k(R-L)} \\
    &= \frac{k(R-L) - (k-1)(R-L)}{k(R-L)} \\
    &= \frac{1}{k}.
\end{align*}

\item If $x_i \geq R$, since $d(i,L)\leq 1-L$ and $d(i,R) \leq 1-R$, again using the definition of $p_{1,L}$, we have
\begin{align*}
    u_i(\bp)
    &= p_{1,L} \cdot \big(1-d(i,L)\big) + \big(1-p_{1,L}\big)\cdot \big(1-d(i,R)\big) \\
    &\geq p_{1,L} \cdot L + \big(1-p_{1,L}\big) \cdot R \\
    &= \frac{(1-k+kR) L + (k-1-kL) R }{k(R-L)} \\
    &= \frac{k-1}{k}.
\end{align*}

\item If $L < x_i < R$, since $d(i,L)=x_i-L$ and $d(i,R) = R-x_i$, we have
\begin{align*}
     u_i(\bp)
    &= p_{1,L} \cdot \big(1-d(i,L)\big) + \big(1-p_{1,L}\big)\cdot \big(1-d(i,R)\big)) \\
    &= p_{1,L}\cdot (1-x_i+L) + \big(1-p_{1,L}\big) \cdot (1-R+x_i) \\
    &= \frac{(1-k+kR)(1-x_i+L) + (k-1-kL)(1-R+x_i) }{k(R-L)} \\
    &= \frac{R-(2k-1)L+2kLR + x_i \cdot (2(k-1) - k(R+L))}{k(R-L)}.
\end{align*}
If $2(k-1) - k(R+L) \geq 0$, then, since $R \geq \frac{k-1}{k}$, 
\begin{align*}
     u_i(\bp) 
    &\geq \frac{R-(2k-1)L+2kLR}{k(R-L)} \geq \frac{1}{k}.
\end{align*}
Otherwise, since $x_i \leq R \leq 1$, we have
\begin{align*}
     u_i(\bp)
    &\geq \frac{R-(2k-1)L+2kLR + R \cdot (2(k-1) - k(R+L))}{k(R-L)} \\
    &= \frac{(2k-1)(R-L) - kR(R-L)}{k(R-L)} \\
    &= \frac{2k-1-kR}{k}\\
    &\geq \frac{k-1}{k}.
\end{align*} 
\end{itemize}
Since $\frac{k-1}{k} \geq \frac{1}{k}$ for any $k \geq 2$, we have that any agent in $N_1$ achieves an expected utility of at least $1/k$, which implies that $\mech \geq n_1/k$. Since $\opt \leq \max_j n_j =n_1$, we have that the approximation ratio is at most $k$. 

\medskip
\noindent
{\bf (Case 4): $\frac{k-1}{k} < L < R \leq 1$.}
Recall that in this case, the mechanism places each facility $F_j$ at $L$ with an appropriately defined probability $p_{j,L}$. Let $F_{j^*}$ be the facility that is placed at a location $O \in [L,R]$ according to an optimal solution. We make the following observations:
\begin{itemize}
    \item For any agent $i \in N_{j^*}$:
    \begin{itemize}
        \item If $x_i > L$, we have that $d(i,L) \leq 1-L$, and thus $1-d(i,L) \geq L \geq L \cdot \big(1-d(i,O)\big)$.
        \item If $x_i \leq L$, we have that $1-d(i,L) \geq 1-d(i,O) \geq L \cdot \big( 1-d(i,O) \big)$.
    \end{itemize}
    \item For any agent $i \in N_j$ such that $j \neq j^*$:
    \begin{itemize}
        \item If $x_i \leq L$, we have that $d(i,L) \leq L$, and thus $1-d(i,L) \geq 1-L$.
        \item If $x_i > L$, we have that $d(i,L) \leq 1-L \leq L$ since $L \geq \frac{k-1}{k} \geq \frac12$ for any $k \geq 2$. Hence, $1-d(i,L) \geq 1-L$. 
    \end{itemize}
\end{itemize}
Using these observations, we can derive the following lower bound on the expected social welfare of the mechanism:
\begin{align*}
    \mech 
    &= p_{j^*,L} \cdot \sum_{i \in N_{j^*}} {\big(1-d(i,L)\big)}  
    + \sum_{j \neq j^*} p_{j,L} \cdot \sum_{i \in N_j}{\big(1-d(i,L)\big)} \\
    &\geq p_{j^*,L} \cdot L \cdot \sum_{i \in N_{j^*}} {\big(1-d(i,O)\big)}  
    + \sum_{j \neq j^*} p_{j,L} \cdot (1-L) n_j \\
    &= p_{j^*,L} \cdot L \cdot \opt
    + (1-L) \cdot \sum_{j \neq j^*} p_{j,L} \cdot n_j.
\end{align*}
If $j^*=1$, since $\opt \leq n_1$, we have that the approximation ratio is
\begin{align*}
    \frac{\opt}{\mech} 
    &\leq \frac{\opt}{p_{j^*,L} \cdot L \cdot \opt
    + (1-L) \cdot \sum_{j \neq j^*} p_{j,L} \cdot n_j} \\
    &\leq \frac{n_1}{p_{j^*,L} \cdot L \cdot n_1
    + (1-L) \cdot \sum_{j \neq j^*} p_{j,L} \cdot n_j}.
\end{align*}
Hence, by the definition of $p_1$, the approximation ratio is at most $k$. 

We next consider the case $j^* \neq 1$. 
Since $\sum_{j \neq 1,j^*} n_j = n-n_1-n_{j^*}$, we have
\begin{align*}
    \mech 
    &\geq \frac{1-p_{1,L}}{k-1} L \cdot \opt + p_{1,L} (1-L) \cdot n_1 + \frac{1-p_{1,L}}{k-1}(1-L)\cdot (n-n_1-n_{j^*}).
\end{align*}
Using this and also the fact that $\opt \leq n_{j^*}$, the approximation ratio is
\begin{align*}
\frac{\opt}{\mech}&\leq 
\frac{n_{j^*}}{\frac{1-p_{1,L}}{k-1} L \cdot n_{j^*} + p_{1,L} (1-L) \cdot n_1 + \frac{1-p_{1,L}}{k-1}(1-L)\cdot (n-n_1-n_{j^*})} \\
&=\frac{n_{j^*}}{\frac{(1-p_{1,L})(2L-1)}{k-1}\cdot n_{j^*} + \frac{(k p_{1,L} -1)(1-L)}{k-1}\cdot n_1 + \frac{(1-p_{1,L})(1-L)}{k-1}\cdot n} \\
&=\frac{(k-1) \cdot n_{j^*}}{(1-p_{1,L})(2L-1)\cdot n_{j^*} + (k p_{1,L} -1)(1-L)\cdot n_1 + (1-p_{1,L})(1-L)\cdot n}.
\end{align*}
If $n_1 > n/2$, then, by definition, $p_{1,L}> 1/k$.
Since the factor of $n_1$ is non-negative in the above upper bound on the approximation ratio, using the facts that $n_1 \geq n_{j^*}$ and $n\geq n_1 + n_{j^*} \geq 2n_{j^*}$, we have that
\begin{align*}
  \frac{\opt}{\mech} &\leq  \frac{k-1}{(1-p_{1,L})(2L-1) + (k p_{1,L} -1)(1-L) + 2(1-p_{1,L})(1-L)} \\
    &= \frac{k-1}{(k-kL-1)p_{1,L} + L} \\
    &< k \cdot \frac{k-1}{k-kL-1 +kL} \\
    &= k.
\end{align*}
Otherwise, when $n_1 \leq n/2$, let $\lambda = n_1/n \leq 1/2$ and observe that the approximation ratio is an increasing function in terms of $n_{j^*} \leq n_1$. Hence,
\begin{align*}
    \frac{\opt}{\mech} &\leq 
    \frac{(k-1)\cdot \lambda}{(1-p_{1,L})(2L-1)\cdot \lambda + (k p_{1,L} -1)(1-L)\cdot \lambda + (1-p_{1,L})(1-L)}.
\end{align*}
By definition, we have
$$p_{1,L} = \frac{ \lambda(k-1) - k(1-L)(1-\lambda)}{k \cdot ( (k-1)L\lambda - (1-L)(1-\lambda))},$$
$$1-p_{1,L} = \frac{ \lambda(k-1)(kL-1)}{k \cdot ( (k-1)L\lambda - (1-L)(1-\lambda))},$$
and
$$kp_{1,L}-1 = \frac{k(k-1)(1-L)(2\lambda-1)}{k \cdot ( (k-1)L\lambda - (1-L)(1-\lambda))}.$$
Hence, the approximation ratio is at most
\begin{align*}
   &k \cdot \frac{\lambda(k-1) \cdot \bigg( (k-1)L\lambda - (1-L)(1-\lambda) \bigg)}
   {(k-1)(kL-1)(2L-1)\lambda^2 + k(k-1)(2\lambda-1)(1-L)^2\lambda + \lambda(k-1)(kL-1)(1-L)} \\
   &= k \cdot \frac{(k-1)L\lambda - (1-L)(1-\lambda)}
   {(kL-1)(2L-1)\lambda + k(2\lambda-1)(1-L)^2 + (kL-1)(1-L)} \\
   &= 
    k \cdot \frac{ \lambda \cdot \bigg( (k-1)L+1-L \bigg) - (1-L)}
   {\lambda \cdot \bigg( (kL-1)(2L-1) + 2k(1-L)^2 \bigg) + \bigg(2kL -k-1\bigg) (1-L)}.
\end{align*}
Since the last expression is an increasing function of $\lambda\leq 1/2$, it is maximized to 
\begin{align*}
   & k \cdot \frac{ \frac12 \cdot \bigg( (k-1)L+1-L \bigg) - (1-L)}
   {\frac12 \cdot \bigg( (kL-1)(2L-1) + 2k(1-L)^2 \bigg) + \bigg(2kL -k-1\bigg) (1-L)} \\
   &= k \cdot \frac{ kL - 1 }{ kL-1 } \\
   &= k.
\end{align*}
The proof is now complete.
\end{proof}

\section{Known Positions} \label{sec:positions}
In this section we consider a restricted setting in which the positions of the agents are assumed to be known and the agents can only misreport about their preferences over the facilities. For a mechanism to be strategyproof in this setting, it must be the case that, for any two instances $I = (\bx,\balpha,C)$ and $I' = (\bx,(\alpha_i',\balpha_{-i}),C)$ which differ only on the preferences reported by a single agent $i$, $u_i(\bp_M(I)) \geq u_i(\bp_M(I'))$. In contrast to the general setting considered in the previous section, where no deterministic strategyproof mechanism can achieve a bounded approximation ratio, we show here that a small constant approximation ratio of $2.325$ is achievable via a deterministic mechanism for any $k \geq 2$. We further improve this bound to $2$ for $k=2$, and complement these upper bounds with a nearly tight lower bound of $3/2$ that holds for any deterministic mechanism and $k \geq 2$. 

For any $k \geq 2$, we consider a mechanism that outputs a different solution depending on how the candidate locations are distributed in the interval $[0,1]$ as a function of a parameter $\theta \in [0,1/2]$ that is given as input. In particular, the mechanism considers the following exhaustive cases:
\begin{itemize}
    \item (Case 1) There exists a candidate location in the interval $[\theta, 1-\theta]$. Then, we pick any such location (for example the one closest to $1/2 \in [\theta, 1-\theta]$) and place there the facility that maximizes the social welfare for that location. 
    
    \item (Case 2) All candidate locations are in the same subinterval among $[0,\theta)$ and $(1-\theta, 1]$. Then, similarly to (Case 1), we pick the location that is closest to $1/2$ and place there the facility that maximizes the social welfare for that location. 
   
    \item (Case 3) There is at least one candidate location in each of $[0,\theta)$ and $(1-\theta, 1]$. 
    Let $c_1$ be the rightmost candidate location in $[0, \theta)$ and $c_2$ be the leftmost one in $(1-\theta,1]$. 
    We identify the welfare-maximizing facility $f_1$ at $c_1$ for the agents in $S_<=\{i:x_i \leq (c_1+c_2)/2 \}$ and the welfare-maximizing facility $f_2$ at $c_2$ for the agents in $S_> = \{i:x_i > (c_1+c_2)/2 \}$.
    Between the two possible solutions $(f_1,c_1)$ and $(f_2,c_2)$, we choose the one with maximum social welfare for the agents in the corresponding subsets.
\end{itemize}
See Mechanism~\ref{mech:locations:upper:k} for a description of the mechanism using pseudocode.

\SetCommentSty{mycommfont}
\begin{algorithm}[h]
\SetNoFillComment
\caption{}
\label{mech:locations:upper:k}
{\bf Input:} $k \geq 2$, candidate locations $C$, known agent positions, reported agent preferences, parameter $\theta \in [0,1/2]$\;
{\bf Output:} Solution $\bw$\;
\tcp*[h]{(Case 1) and (Case 2)} \\
\uIf{$C \cap [\theta,1-\theta] \neq \varnothing$ or $C \subseteq [0,\theta)$ or $C \subseteq (1-\theta,1]$}
{
   $c \gets $ candidate location closest to $1/2$\;
   $f \gets \arg\max_{j \in [k]} \SW(j,c)$\;
   $\bw \gets (f,c)$\;
}
\tcp*[h]{(Case 3)} \\
\uElse{
    $c_1 \gets $ rightmost candidate location in $[0,\theta)$\;
    $c_2 \gets $ leftmost candidate location in $(1-\theta,1]$\;
    $S_< \gets \{i: x_i\leq (c_1+c_2)/2\}$\;
    $S_> \gets \{i: x_i>(c_1+c_2)/2\}$\;
    $f_1 \gets \arg\max_{j \in [k]} \sum_{i \in N_j \cap S_<} u_i(j,c_1)$\;
    $f_2 \gets \arg\max_{j \in [k]} \sum_{i \in N_j \cap S_>} u_i(j,c_2)$\;
    \If{$ \sum_{i \in N_{f_1} \cap S_<} u_i(f_1,c_1) \geq  \sum_{i \in N_{f_2} \cap S_>} u_i(f_2,c_2)$}
    {$\bw \gets (f_1,c_1)$\;}
    \Else
    {$\bw \gets (f_2,c_2)$\;}
}   
\Return $\bw$\;
\end{algorithm}

We first argue that the mechanism is strategyproof. 

\begin{theorem}\label{thm:locations:upper:k:sp}
When the positions of the agents are known, Mechanism~\ref{mech:locations:upper:k} is strategyproof.
\end{theorem}

\begin{proof}
First observe that all cases considered by the mechanism depend only on how the candidate locations are distributed in the interval $[0,1]$, and thus no agent can manipulate the mechanism to switch from one case to another.

\medskip
\noindent
{\bf (Case 1) and (Case 2)}. The location $c$ is fixed and cannot be affected by the reported preferences of the agents. All agents that truly approve the facility $F_f$ that maximizes the social welfare for $c$ do not have any incentive to misreport their preferences as they achieve the maximum possible utility they can. Finally, any agent $i$ that does not truly approve $F_f$ cannot change the outcome in its favor; by misreporting, $i$ can only lead to an increased social welfare for a facility that $i$ does not approve or a decreased social welfare for a facility that $i$ does approve, thus again leading to facility that $i$ does not approve to be placed at $c$.  

\medskip
\noindent
{\bf (Case 3)}. The locations $c_1$ and $c_2$ are fixed and cannot be affected by the reported preferences of the agents. Every agent of $i \in S_<$ prefers a facility it approves to be placed at $c_1$ and can only affect this choice for $c_1$ (not $c_2$). In particular, by misreporting, $i$ can either decrease the total utility of the agents in $S_<$ for the facilities that $i$ approves or increase the total utility of the agents in $S_<$ for the facilities that $i$ does not approve. These changes may not affect the outcome at all, or lead to a worse outcome (a facility that $i$ approves to be placed at $c_2$ rather than $c_1$, or a facility that $i$ does not approve is chosen). In any case, $i$ has no incentive to misreport. The case of agents in $S_>$ is similar.
\end{proof}

We next show an upper bound on the approximation ratio of the mechanism as a function of the parameter $\theta$. We will later optimize over $\theta$ to obtain a bound of approximately $2.325$. 

\begin{theorem}\label{thm:locations:upper:approx}
The approximation ratio of mechanism \ref{mech:locations:upper:k} is at most $\max\left\{\frac{1}{\theta}, 1-\theta+\frac{1}{1-\theta}\right\}$, for any $\theta \in [0,1/2]$.
\end{theorem}

\begin{proof}
Without loss of generality, assume that the optimal solution is $\bo = (1,y)$, that is, facility $F_1$ is placed at some candidate location $y$. We consider each case of the mechanism separately. In (Case 1) and (Case 2) we will show an upper bound of $1/\theta$, and in (Case 3) we will show an upper bound of $1-\theta+\frac{1}{1-\theta}$. 

\medskip

\noindent
{\bf (Case 1) There exists a candidate location in $[\theta, 1-\theta]$}.
Let $c$ be the location in $[\theta,1-\theta]$ chosen by the mechanism. Since the facility placed at $c$ is the one that maximizes the social welfare given $c$ as the location, we have that $\mech \geq \SW(1,c)$. 
For any agent $i \in N_1$, $d(i,c) \leq 1-\theta$, and thus $u_i(1,c) = 1 - d(i,c) \geq \theta \geq \theta \cdot u_i(\bo)$.
Hence, $\SW(1,c) \geq \theta \cdot \opt$, which implies an approximation ratio of at most $1/\theta$. 

\medskip
\noindent 
{\bf (Case 2) All candidate locations are in the same subinterval among $[0,\theta)$ and $(1-\theta,1]$.}
Due to symmetry, and since $1-\theta \geq \theta$ by definition, it suffices to consider only the first subcase where all candidate locations are in $[0,\theta]$. Then, since the location $c$ that is chosen by the mechanism is the closest one to $1/2$, it has to be the rightmost location. For any agent $i \in N_1$, we have:
\begin{itemize}
    \item If $x_i \geq \theta$, then $u_i(1,c) \geq u_i(\bo) \geq \theta \cdot u_i(\bo)$.
    \item If $x_i < \theta$, then $d(i,c) \leq \theta$, and thus $u_i(1,c) = 1 - d(i,c) \geq 1-\theta \geq \theta \geq \theta \cdot u_i(\bo)$.  
\end{itemize}
Since the facility that the mechanism places at $c$ is the welfare-maximizing one, as in (Case 1), we have that
$\mech \geq \SW(1,c) \geq \theta \cdot \opt$, that is, the approximation ratio is at most $1/\theta$.

\medskip
\noindent 
{\bf (Case 3) At least one facility is in $[0,\theta)$ and at least one is in $(1-\theta,1]$.}
Assume that the solution computed by the mechanism is $\bw = (f_1,c_1)$; the analysis is similar for the remaining case where a facility is placed at $c_2$. Recall that $S_< = \{i: x_i\leq (c_1+c_2)/2\}$ and $S_> = \{i: x_i>(c_1+c_2)/2\}$.
By the definition of the mechanism, solution $\bw = (f_1,c_1)$ maximizes the total utility of the agents in $S_<$ over all solutions $(j,c_1)$. Hence, 
\begin{align}\label{eq:locations:left-optimal}
\sum_{i \in N_{f_1} \cap S_<}u_i(\bw) \geq \sum_{i \in N_1 \cap S_<} u_i(1,c_1).
\end{align}
In addition, since $\bw$ is chosen by the mechanism and solution $(f_2,c_2)$ maximizes the total utility of the agents in $S_>$ over all solutions $(j,c_2)$, we also have
\begin{align}\label{eq:locations:left-more-than-right}
\sum_{i \in N_{f_1} \cap S_<}u_i(\bw) \geq \sum_{i \in N_{f_2} \cap S_>} u_i(f_2,c_2) \geq \sum_{i \in N_1 \cap S_>}u_i(1,c_2).
\end{align}
Since $c_1$ is the rightmost location in $[0,\theta]$ and $c_2$ is the leftmost location in $[1-\theta,1]$, it suffices to argue about two subcases depending on the optimal location $y$ relative to $c_1$ and $c_2$: (i) $y\leq c_1$, and (ii) $y\geq c_2$. 
In each subcase, we will bound from above the optimal social welfare by separately bounding the utility of the agents in $N_1 \cap S_<$ and the utility of the agents in $N_1 \cap S_>$.

\medskip
\noindent 
{\bf Subcase (i) $y\leq c_1$.}
We have:
\begin{itemize}
\item 
For any agent $i \in N_1 \cap S_<$ such that $x_i \in \left[ \frac{y+c_1}{2}, \frac{c_1+c_2}{2} \right]$, $u_i(1,c_1)\geq u_i(\bo)$. 
For any agent $i \in N_1 \cap S_<$ such that $x_i \leq (y+c_1)/2$, we claim that
$u_i(1,c_1) \geq \frac{1-c_1}{1-y} \cdot u_i(\bo)$. Indeed:
\begin{itemize}
    \item If $x_i \leq y$, then    
    \begin{align*}
        u_i(1,c_1) \geq \frac{1-c_1}{1-y} \cdot u_i(\bo) 
        &\Leftrightarrow 1-c_1+x_i \geq \frac{1-c_1}{1-y}(1-y+x_i) \\
        &\Leftrightarrow x_i (c_1-y)\geq 0.
    \end{align*}
    \item If $x_i \in \left( y, \frac{y+c_1}{2} \right]$, then 
    \begin{align*}
         u_i(1,c_1) \geq \frac{1-c_1}{1-y} \cdot u_i(\bo) 
         &\Leftrightarrow 1-c_1+x_i \geq \frac{1-c_1}{1-y}(1-x_i+y) \\
         &\Leftrightarrow x_i \geq \frac{2y(1-c_1)}{2-c_1-y}.
    \end{align*}
    The last inequality is true for any $x_i \geq y$, since $y \geq \frac{2y(1-c_1)}{2-c_1-y} \Leftrightarrow y(c_1-y) \geq 0$. 
\end{itemize}
By the above observation about the utility of these agents, and using \eqref{eq:locations:left-optimal}, we obtain
\begin{align*}
    \sum_{i\in N_1 \cap S_<} u_i(\bo)
    &\leq \frac{1-y}{1-c_1} \cdot \sum_{N_1 \cap S_<} u_i(1,c_1)\\
    &\leq \frac{1-y}{1-c_1} \cdot \sum_{N_{f_1} \cap S_<} u_i(\bw).
\end{align*}

\item
For any agent $i \in N_1 \cap S_>$, since $d(i,c_2) \leq d(i,c_1)$, we have that $d(i,y) = d(i,c_1) + d(c_1,y) \geq d(i,c_2) + d(c_1,y)$, and thus 
$$u_i(\bo) = 1-d(i,y) \leq 1-d(i, c_2)-d(y,c_1) = u_i(1,c_2) +y-c_1.$$
Using this, the fact that $|N_1 \cap S_>| \geq \sum_{i \in N_1 \cap S_>} u_i(1,c_2)$, and \eqref{eq:locations:left-more-than-right}, we obtain
\begin{align*}
    \sum_{i\in N_1 \cap S_>} u_i(\bo)
    &\leq \sum_{i\in N_1 \cap S_>}\bigg( u_i(1,c_2) +y-c_1 \bigg)\\
    &= \sum_{i\in N_1 \cap S_>} u_i(1,c_2) + (y-c_1) \cdot |N_1 \cap S_>| \\
    &\leq (1+y-c_1) \cdot \sum_{i\in N_1 \cap S_>} u_i(1,c_2) \\
    &\leq (1+y-c_1)\cdot \sum_{i \in N_{f_1} \cap S_<}u_i(\bw).
\end{align*}    
\end{itemize}
By putting everything together, we have
\begin{align*}
    \opt &= \sum_{i\in N_1 \cap S_<} u_i(\bo) + \sum_{i\in N_1 \cap S_>} u_i(\bo) \\
    &\leq \bigg(\frac{1-y}{1-c_1} + 1+y-c_1 \bigg)\cdot \sum_{i \in N_{f_1} \cap S_<}u_i(\bw) \\
    &\leq \bigg(\frac{1-y}{1-c_1} + 1+y-c_1 \bigg)\cdot \mech.
\end{align*}
As $y\leq c_1$, the approximation ratio $\frac{1-y}{1-c_1} + 1+y-c_1 $ is maximized for $y=0$ and becomes at most $\frac{1}{1-c_1} + 1 - c_1$, which in turn is maximized for $c_1=\theta$ and becomes at most $1/(1-\theta) + 1-\theta$.

\medskip
\noindent 
{\bf Subcase (ii) $y\geq c_2$.}
We have:
\begin{itemize}
    \item
For any agent $i \in N_1 \cap S_<$, since $d(i,y) = d(i,c_2) + d(c_2,y)$ and $d(i,c_1) \leq d(i,c_2)$, 
$$u_i(\bo) = 1-d(i,y) \leq 1-d(i, c_1)-d(c_2,y) = u_i(1,c_1) +c_2-y.$$
Using this, the fact that $|N_1 \cap S_<| \geq \sum_{i \in N_1 \cap S_<}u_i(1,c_1)$, and \eqref{eq:locations:left-optimal}, we obtain
\begin{align*}
    \sum_{i\in N_1 \cap S_<} u_i(\bo)
    &\leq \sum_{i\in N_1 \cap S_<}\bigg( u_i(1,c_1) +c_2-y \bigg)\\
    &= \sum_{i\in N_1 \cap S_<} u_i(1,c_1) + (c_2-y) \cdot |N_1 \cap S_<| \\
    &\leq (1+c_2-y) \cdot \sum_{i\in N_1 \cap S_<} u_i(1,c_1) \\
    &\leq (1+c_2-y)\cdot \sum_{i \in N_{f_1} \cap S_<}u_i(\bw).
\end{align*}

\item 
For any agent $i \in N_1 \cap S_>$ such that $x_i \in \left[\frac{c_1+c_2}{2}, \frac{c_2+y}{2}\right]$, $u_i(1,c_2) \geq u_i(\bo)$. 
For any agent $i \in N_1 \cap S_>$ such that $x_i > (c_2+y)/2$, we claim that 
$u_i(1,c_2) \geq \frac{c_2}{y} \cdot u_i(\bo)$.
Indeed: 
\begin{itemize}
    \item If $x_i \geq y$, since $y\geq c_2$, it holds 
    \begin{align*}
    u_i(1,c_2) \geq \frac{c_2}{y}\cdot u_i(\bo)
    &\Leftrightarrow 1-x_i+c_2 \geq \frac{c_2}{y}\cdot (1-x_i+y) \\
    &\Leftrightarrow (y-c_2)(1-x_i) \geq 0.
    \end{align*}
    
    \item If $x_i \in \left(\frac{c_2+y}{2},y\right]$, then 
    \begin{align*}
    u_i(1,c_2) \geq \frac{c_2}{y}\cdot u_i(\bo)
    &\Leftrightarrow 1-x_i+c_2 \geq \frac{c_2}{y}\cdot (1-y+x_i) \\
    &\Leftrightarrow x_i \leq \frac{y+2c_2 y-c_2}{y+c_2}.
    \end{align*}
    The last inequality is true for any $x_i \leq y$ since $ y \leq \frac{y+2c_2 y-c_2}{y+c_2} \Leftrightarrow (y-c_2)(1-y)\geq 0$.
\end{itemize}
By the above observation about the utility of these agents, and using \eqref{eq:locations:left-more-than-right}, we obtain
\begin{align*}
 \sum_{i\in N_1 \cap S_>} u_i(\bo) 
&\leq \frac{y}{c_2}\cdot \sum_{N_1 \cap S_>} u_i(1,c_2)\\
&\leq \frac{y}{c_2}\cdot \sum_{N_1 \cap S_<} u_i(\bw).
\end{align*}
\end{itemize}
By putting everything together, we have
\begin{align*}
    \opt &= \sum_{i\in N_1 \cap S_<} u_i(\bo) + \sum_{i\in N_1 \cap S_>} u_i(\bo) \\
    &\leq \bigg(1+c_2-y + \frac{y}{c_2} \bigg)\cdot \sum_{i \in N_{f_1} \cap S_<}u_i(\bw) \\
    &\leq \bigg(1+c_2-y + \frac{y}{c_2} \bigg)\cdot \mech.
\end{align*}
As $y\geq c_2$, the approximation ratio $1+c_2-y + \frac{y}{c_2}$ is maximized for $y=1$ and becomes at most $c_2+1/c_2$ which in turn is maximized for $c_2=1-\theta$ and becomes at most $1-\theta + 1/(1-\theta)$.
\end{proof}

By appropriately tuning parameter $\theta$, we obtain different approximation guarantees, for example, by setting $\theta = 1/2$, we get an approximation ratio of at most $\max\{2,5/2\} = 5/2$. By balancing out the two terms, the best possible approximation ratio is at most $2.325$ and is achieved for $\theta \approx 0.43$. 

\begin{corollary}\label{cor:locations:upper:k:optimization}
For $\theta = \frac13 \left(2 - 5 \sqrt[3]{\frac{2}{3 \sqrt{69}-11}} + \sqrt[3]{\frac12 \left(3 \sqrt{69}-11\right)}\right) \approx 0.43$, the approximation ratio of mechanism \ref{mech:locations:upper:k} is at most $ 1/\theta \leq 2.325$. 
\end{corollary}

For $k=2$, we can achieve an improved upper bound of $2$ by a much simpler mechanism, which first chooses the candidate location $c$ that minimizes the total distance from all agents, and then places at $c$ the facility $f$ that maximizes the total utility of the agents that approve it. See Mechanism~\ref{mech:locations:upper:2}.

\SetCommentSty{mycommfont}
\begin{algorithm}[h]
\SetNoFillComment
\caption{}
\label{mech:locations:upper:2}
{\bf Input:} $k = 2$, candidate locations $C$, known agent positions, reported agent preferences\;
{\bf Output:} Solution $\bw$\;
$c \gets \arg\min_{x \in C} \sum_i d(i,x)$\;
$f \gets \arg\max_{j \in [k]} \sum_i u_i(j,c)$\;
$\bw \gets (f,c)$\;
\Return $\bw$\;
\end{algorithm}

\begin{theorem}\label{thm:locations:upper:approx:2}
When the positions of the agents are known and $k=2$, Mechanism~\ref{mech:locations:upper:2} is strategyproof and achieves an approximation ratio of at most $2$.
\end{theorem}

\begin{proof}
First observe that location $c$ is chosen according to the known information about the positions of the agents, and thus it cannot be affected by any possible misreport of the agents about their preferences. Since the facility chosen to be placed at $c$ maximizes the social welfare of the agents, the agents that approve it clearly have no incentive to misreport. By misreporting, any agent $i$ that does not approve the chosen facility, can only increase the social welfare of the facility that it does not approve and possibly decrease the social welfare of the facility it does approve; hence, the outcome may either not change or become even worse, and thus $i$ also has no incentive to misreport. 

For the approximation ratio, without loss of generality, let $\bo = (1,y)$ be an optimal solution, that is, facility $F_1$ is placed at some location $y$. By the definition of the mechanism, the solution $\bw = (f,c)$ satisfies the inequality:
\begin{align*}
    \sum_{i \in N_f} \bigg( 1 - d(i,c) \bigg) \geq \sum_{i \in N_{3-f}} \bigg( 1 - d(i,c) \bigg),
\end{align*}
which implies that
\begin{align*}
    \sum_{i \in N_f} \bigg( 1 - d(i,c) \bigg) \geq \frac12 \cdot \sum_{i \in N} \bigg( 1 - d(i,c) \bigg).
\end{align*}
Since $c$ minimizes the total distance of all agents, we further obtain that
\begin{align*}
    \sum_{i \in N_f} \bigg( 1 - d(i,c) \bigg) \geq \frac12 \cdot \sum_{i \in N} \bigg( 1 - d(i,y) \bigg).
\end{align*}
Since
\begin{align*}
    \opt = \sum_{i \in N_1} \bigg( 1 - d(i,y) \bigg) \leq \sum_{i \in N} \bigg( 1 - d(i,y) \bigg), 
\end{align*}
we finally obtain
\begin{align*}
    \mech 
    &= \sum_{i \in N_f} u_i(f,c) \\
    &= \sum_{i \in N_f} \bigg( 1 - d(i,c) \bigg) \\
    &\geq \frac12 \cdot \sum_{i \in N_1} \bigg( 1 - d(i,y) \bigg) \\
    &= \frac12 \cdot \opt.
\end{align*}
Hence, the approximation ratio is at most $2$.
\end{proof}

It is not hard to see that Mechanism~\ref{mech:locations:upper:2} can be generalized to work for any $k\geq 2$, but it leads to a tight $k$-approximation, which is better than the approximation of $2.325$ achieved by Mechanism~\ref{mech:locations:upper:k} only for $k=2$.

We complement the above positive results with an impossibility: No deterministic strategyproof mechanism can achieve an approximation ratio better than $3/2$, for any $k \geq 2$. 

\begin{theorem}\label{thm:locations:lower:det}
When the positions of the agents are known, the approximation ratio of any deterministic strategyproof mechanism is at least $3/2$, for any $k \geq 2$.
\end{theorem}

\begin{proof}
Let $\varepsilon > 0$ be an infinitesimal. 
Consider an instance $I$ with two candidate locations at $0$ and $1$. There is an agent at $0$ that approves only $F_1$, two agents at $1/2-\varepsilon$ that approve all facilities, two agents at $1/2+\varepsilon$ that also approve all facilities, and an agent at $1$ that approves only $F_2$. Treating $\varepsilon$ as $0$ to simplify calculations, the social welfare of solutions $(1,0)$ and $(2,1)$ is $3$, whereas the social welfare of any other solution is at most $2$. Hence, if any non-optimal solution is chosen, the approximation ratio is at least $3/2$. So, suppose that the solution chosen by the mechanism for $I$ is $(2,1)$, that is, $F_2$ is placed at $1$; the case where the solution is $(1,0)$ is symmetric.

We now change the preferences of the agents at $1/2-\varepsilon$ one-by-one so that they approve all facilities but $F_2$. Since the mechanism is strategyproof, in the new instances derived after each such change, a facility (not necessarily $F_2$) must still be placed at $1$; otherwise, the deviating agent would increase its utility by $\varepsilon > 0$. In the last instance, the maximum possible social welfare by placing a facility at $1$ is $2$ (in particular, this is achieved by placing any facility different than $F_2$ at 1), whereas the optimal social welfare is $3$ and is achieved by placing $F_1$ at $0$. Hence, the approximation ratio is at least $3/2$. 
\end{proof}

We conclude with a lower bound of $6/5$ on the approximation ratio of randomized mechanisms, which leaves open the possibility of achieving improved guarantees by exploiting randomization, but is a quite challenging task. 

\begin{theorem}\label{thm:locations:lower:rand}
When the positions of the agents are known, the approximation ratio of any randomized strategyproof mechanism is at least $6/5$, for any $k \geq 2$.
\end{theorem}

\begin{proof}
We consider the same set of instances as in the proof of Theorem~\ref{thm:locations:lower:det}. We again start from instance $I$, in which there are two candidate locations at $0$ and $1$, an agent at $0$ that approves only facility $F_1$, two agents at $1/2-\varepsilon$ that approve all facilities, two agents at $1/2+\varepsilon$ that approve all facilities, and an agent at $1$ that approves only $F_2$. Without loss of generality, let $p \geq 1/2$ be the probability with which the mechanism places a facility at $0$ when given $I$ as input.

Consider now instance $J$ where, one-by-one, the two agents at $1/2 - \varepsilon$ change their preference so that they approve all facilities but $F_2$. Due to strategyproofness, in each change, the probability with which the mechanism chooses $1$ as the location to place a facility must be $p' \geq p \geq 1/2$; otherwise, the probability of assigning a facility at $0$ would increase and thus the deviating agent (who starts from approving all facilities) would increase its expected utility. In instance $J$, the optimal solution is to place $F_1$ at $0$ with a social welfare of $3$, whereas the social welfare of any other solution is at most $2$. Hence, the expected social welfare of the mechanism is at most $3(1-p')+2p' \leq 5/2$, and the approximation ratio is at least $6/5$. 
\end{proof}

\section{Conclusion and Open Problems} \label{sec:open}
In this work, we considered a truthful facility location problem with sufficient funds to only build one out of $k$ available facilities at a location chosen from a set of candidate ones, aiming to (approximately) maximize the social welfare of the agents. For the general setting, where agents may misreport both their positions and their preferences over the facilities, we showed that the approximation ratio of deterministic strategyproof mechanisms is unbounded, whereas that of randomized mechanisms is $k$. For the restricted setting of known positions, we showed that it is possible to achieve a small constant approximation ratio of nearly $2.325$ using a deterministic mechanism, and complemented this result with lower bounds of $3/2$ and $6/5$ for deterministic and randomized mechanisms, respectively. Closing these gaps is the most challenging question that our work leaves open. 

There are multiple interesting generalizations and extensions to explore in the future. 
Our model can be thought of as combining elements of the classic facility location problem and single-winner voting in the sense that we aim to choose a location to build one facility which is chosen from a set of $k$ available options. Consequently, it would make sense to also consider the case of multiwinner voting where $\ell \geq 2$ facilities out of the $k$ available ones can be chosen. For this model, there are several ways in which the individual utility of the agents can be defined, depending on whether they are affected by all their approved facilities that are chosen to be built or just a few of them, similarly to the min-, sum- and max-variants that have been considered in models where all facilities can be built. Taking this a step further, it would then be interesting to generalize the setting even more to the case of participatory budgeting~\citep{Aziz2021}, where we are given a fixed budget, each facility has a particular cost and we can only build facilities with a total cost that satisfies the budget constraint.


\bibliographystyle{plainnat}
\bibliography{references}

\end{document}